\documentclass[prl,twocolumn,floatfix,superscriptaddress]{revtex4-2}
\usepackage{amsthm,bbm,dsfont}
\usepackage{amsmath}
\usepackage{tikz}
\usepackage{soul}
\usetikzlibrary{calc}
\usepackage[colorlinks,citecolor=blue,linkcolor=blue,urlcolor=blue]{hyperref} 

\newtheorem{theorem}{Theorem}

\newtheorem{definition}[theorem]{Definition}
\newtheorem{lemma}[theorem]{Lemma}

\newcommand{\ghz}{|\mathrm{GHZ}\rangle}
\newcommand{\supp}{\operatorname{supp}}
\newcommand{\meane}[1]{\lceil  #1 \rceil}
\newcommand{\expct}[1]{\langle  #1 \rangle}
\newcommand{\floor}[1]{\lfloor  #1 \rfloor}
\newcommand{\tr}{{\mathrm{Tr}}}
\newcommand{\id}{{\mathds{1}}}
\newcommand{\re}{\operatorname{Re}}
\newcommand{\var}{\operatorname{Var}}

\begin{document}
\title{
Quantum LOSR networks cannot generate graph states with high fidelity
}
\author{Yi-Xuan Wang}
\affiliation{School of Physics and Optoelectronics Engineering, Anhui University, 230601 Hefei, China}
\affiliation{Naturwissenschaftlich-Technische Fakult\"at, Universit\"at Siegen, Walter-Flex-Stra{\ss}e 3, 57068 Siegen, Germany}
\affiliation{Department of Modern Physics, University of Science and Technology of China, Hefei 230026, China}
\affiliation{Institute of Theoretical Physics and IQST, Albert-Einstein Allee 11, Ulm University, 89081 Ulm, Germany}
\author{Zhen-Peng Xu}
\email{zhen-peng.xu@ahu.edu.cn}
\affiliation{School of Physics and Optoelectronics Engineering, Anhui University, 230601 Hefei, China}
\affiliation{Naturwissenschaftlich-Technische Fakult\"at, Universit\"at Siegen, Walter-Flex-Stra{\ss}e 3, 57068 Siegen, Germany}
\author{Otfried G\"{u}hne}
\affiliation{Naturwissenschaftlich-Technische Fakult\"at, Universit\"at Siegen, Walter-Flex-Stra{\ss}e 3, 57068 Siegen, Germany}

\date{\today}

\begin{abstract}
Quantum networks lead to novel notions of locality and correlations and an 
important problem concerns the question of which quantum states can be 
experimentally prepared with a given network structure and devices and which 
not.  We prove that all multi-qubit graph states arising from a connected graph 
cannot originate from any quantum network with bipartite sources, as long as 
feed-forward and quantum memories are not available. Moreover, the fidelity of a 
multi-qubit graph state and any network state cannot exceed $9/10$. Similar 
results can also be established for a large class of multi-qudit graph states. 
\end{abstract}
\maketitle

\textit{Introduction.}---
A central topic in quantum information theory is entanglement 
theory~\cite{horodecki2009quantum}. Since quantum entanglement 
is a valuable resource for quantum metrology~\cite{giovannetti2006quantum}, 
quantum computation~\cite{divincenzo1995quantum,illiano2022quantum}, quantum key distribution~\cite{Bennett1984quantum} and anonymous quantum conference key agreement~\cite{hahn2020anonymous}, bundles of experimental effort have been 
devoted to create more entangled state~\cite{mooney2021generation,pogorelov2021compact,omran2019generation,wang201818}, 
and series of theoretical works have contributed to characterize different types of entanglement~\cite{acin2001classification, vidal2000reversible, dur1999separability}.

A quantum state is said to be entangled if it is not a convex combination 
of product states. To detect the quantum entanglement of two particles,
witnesses~\cite{Horodecki_1996wit,Terhal_2000wit} 
(often related to the fidelity of a target state~\cite{acin2001classification,Weilenmann_2020fied_target,G_hne_2021fied_target}) and the PPT 
criterion~\cite{peres1996separability} are standard methods. 
In the multiparticle case a state can be entangled, but separable 
for a given bipartition, then it is called biseparable~\cite{dur1999separability}. 
If the state cannot be written as a convex combination of any biseparable 
states, it is said to be genuine multiparticle entangled (GME)~\cite{acin2001classification,vidal2000reversible,dur1999separability}.
As for the detection of GME states, the fidelity with a highly entangled state,
such as the Greenberger-Horne-Zeilinger (GHZ) state is one of the most used in experiments~\cite{guhne2010separability,mooney2021generation}. For example, 
an $n$-partite state with a GHZ fidelity exceeding $1/2$ is GME. 
Including GHZ state as a special example, graph states~\cite{audenaert2005entanglement,hein2006entanglement,looi2008quantum} 
play an eminent role in entanglement theory and its applications.

Recently, the concept of genuine multiparticle entanglement has been 
debated \cite{Yamasaki_2022genuine,Palazuelos_2022genuine}, and novel 
notions appropriate for the network scenario have been introduced and studied \cite{navascues2020genuine,luo2021new, hansenne2022symmetries}. In a network,
states can be prepared by distributing particles from multiple 
smaller sources to different parties and applying local channels, 
see Fig.~\ref{fig:TQN} for an example. {In this fundamental
scenario the local operations rely on a globally shared classical variable 
(Local operations and shared randomness, LOSR), e.g., a predefined protocol
with shared randomness. The scenario of local operations assisted by classical 
communication (LOCC) gives more power to create distributed quantum states. But, 
communication based on outcome of local operations requires considerable time, either in the scenario of distributed quantum computation where local operations takes a while, or in the case that nodes in the network are far away from each other. Consequently,
it would also require the usage of quantum 
memories or feed-forward techniques, which are expensive resources for current
quantum technologies. Moreover, device-independent quantum information protocols
are frequently related to Bell scenarios, where communication is impossible due
to the space-like separation.
Overall, an $n$-partite state is called genuine network multipartite entangled (GNME) if it cannot be created via LOSR in the network approach using $(n-1)$-partite sources 
only. 
Besides this, other quantum correlations like quantum nonlocality and quantum steering 
have also been generalized to quantum networks recently~\cite{pozas-kerstjens2019bounding, renou2019genuine, gisin2020constraints, aberg2020semidefinite,tavakoli2021bell, 
kraft2021characterizing, kraft2021quantum, jones2021network}.

\begin{figure}[t]
\includegraphics[width=0.28\textwidth]{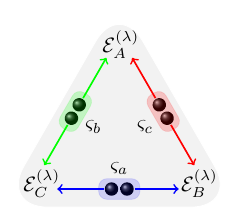}
\caption{Sketch of a triangle quantum network, where three 
bipartite sources $\varsigma_a$, $\varsigma_b$ and $\varsigma_c$ 
are distributed to three parties $A, B, C$, e.g., the source 
states $\varsigma_a$ are sent to $B$ and $C$. Each party can 
apply local operations $\mathcal{E}_X^{(\lambda)}$ ($X=A,B,C$) 
on the received particles, which are affected by global shared 
randomness $\lambda$. 
}
\label{fig:TQN}
\end{figure}  

For three-qubit states in the triangle network, a witness 
derived from the fidelity~\cite{navascues2020genuine, kraft2021quantum} 
and semidefinite programming methods based on the inflation technique~\cite{navascues2020genuine,wolfe2021quantum} can be useful. The disadvantage of
these approaches is that both of them are hard to generalize to more 
complicated quantum networks. An analytical method based on symmetry 
analysis and inflation techniques~\cite{wolfe2019inflation} was proposed 
recently~\cite{hansenne2022symmetries} and can overcome some of the 
difficulties. Explicitly, it was shown there that all $n$-qubit graph 
states with $n \le 12$ are not available in networks with bipartite 
sources,  and it was conjectured that this no-go theorem hold for
multi-qubit graph states with an arbitrary number of particles. 

In this paper 
we prove that all multi-qudit graph states with a connected graph, 
where the multiplicities of the edges are either constant or zero, 
cannot be prepared in any network with only bipartite sources. 
In fact, this result holds also for all the states whose fidelity 
with some of those qudit graph states exceeds a certain value. 
More specifically, our results exclude the generation of any multi-qubit 
graph state with a fidelity larger than $9/10$ in networks. This proves 
the conjecture formulated in Ref.~\cite{hansenne2022symmetries}, it also 
may provide interesting connections to other no-go theorems on the 
preparability of graph states in different physical scenarios, such 
as spin-models with two-body 
interactions~\cite{van2008graph,huber2016characterizing}.

\textit{Network entanglement.}---
The definition of network entanglement is best explained using 
the example of the triangle scenario, see Fig.~\ref{fig:TQN}. Here, one 
has three bipartite quantum source states $\varsigma_x$ for $x=a, b, c$, 
and three local channels $\mathcal{E}_{X}^{(\lambda )}$ for $X=A, B, C$, 
where $\lambda$ is the shared random variable with probability $p_\lambda$. 
The global state can be prepared in this scenario has the form
\begin{equation}
\varrho =\mathop{\sum}\limits_{\lambda }{p}_{\lambda }\mathcal{E}_{A}^{(\lambda )}\otimes \mathcal{E}_{B}^{(\lambda )}\otimes \mathcal{E}_{C}^{(\lambda )}\left[\varsigma_a\otimes \varsigma_b \otimes \varsigma_c\right].
\label{eq-networkdef}
\end{equation}
For a given state $\varrho$, the question arises whether it can be generated
in this manner, and this question was in detail discussed in Refs.~\cite{navascues2020genuine, luo2021new, hansenne2022symmetries}.

More generally, one can introduce quantum network states as follows: 
A given hypergraph $G(V,E)$, where  $V$ is the set of vertices and  
$E$ is the set of hyperedges (i.e., sets of vertices)  describes a network
where each vertex stands for a party and each hyperedge stands for a source 
which dispatches particles to the parties represented by the vertices in it. 
This quantum network is a correlated quantum network (CQN), if each party can 
apply a local channel depending on a shared random variable. Then, in complete
analogy to the definition in Eq.~(\ref{eq-networkdef}) one can ask whether 
a state can be prepared in this network or not. More properties of these
sets of states (e.g., concerning the convex structure and extremal points)
can be found in Refs.~\cite{navascues2020genuine, hansenne2022symmetries}.

\textit{Three-qudit GHZ states and the inflation technique.}---
The inflation technique~\cite{wolfe2019inflation} turns out to be an useful tool to study network entanglement~\cite{navascues2020genuine,hansenne2022symmetries}. Unless otherwise stated,  we consider networks with bipartite sources only. For a given network as the triangle network depicted in Fig.~\ref{fig:TQN} and also in Fig.~2(a), the copies of sources are sent to different copies of the parties as in Fig.~2(b) and Fig.~2(c). In principle, the source states may also be wired differently in different kinds of inflation. For convenience, here we use also edges with only one vertex as in Fig.~2(d) and Fig.~2(e), which means that the particles, which have not appeared in the state, are traced out in the source states. 

The key idea of the inflation method is the following: If the three-particle state $\rho$ can be prepared in the network, the six-particle states $\gamma$ and $\eta$ can be prepared in the inflated
networks as well. The states $\gamma$ and $\eta$ share some marginals 
with $\rho$ and with each other. So, if one can prove that six-particle
states with these desired properties do not exist, the state $\rho$
is not reachable in the original network.

\begin{figure}[t]
\includegraphics[width=0.45\textwidth]{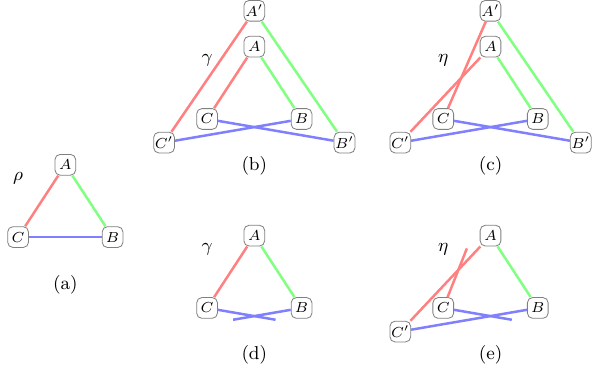}
\caption{The triangle network in (a) and two kinds of its inflation in (b) and (c). For convenience, we use (d) and (e) as short notations of (b) and (c), respectively. Here the broken edge $e$ with only one vertex $v$ means that we ignore or trace out the particles not for the party represented by $v$ from the source state represented by $e$.}
\label{fig:GHZ}
\end{figure}

To see how the idea of inflation works in practice, we 
take the three-qudit GHZ state \cite{greenberger1989going} 
as an example. The three-qudit GHZ state $\ghz = 1/\sqrt{d}\sum_{i=0}^{d-1} |iii\rangle$ is a stabilizer state, {whose stabilizers include}
\begin{align}\label{eq:stablizer_ghz}
    Z_B Z_A^\dagger,\ 
    Z_A Z_C^\dagger,\ 
    Z_B Z_C^\dagger,\ 
    X_A X_B X_C,
\end{align}
where the unitary operators are
$
    Z=\sum_{q=0}^{d-1}\omega^q|q\rangle\langle q|,\ 
    X=\sum_{q=0}^{d-1}|q\oplus 1\rangle\langle q|,
$
with $\oplus$ to be the addition modulo $d$, and $\omega=\exp \frac{2\pi i}{d}$. Note $X^m Z^n=\omega^{-mn}Z^n X^m$.

For a given network state $\rho$, we can consider the value $\meane{M}_\rho$ of any stabilizer $M$ in Eq.~\eqref{eq:stablizer_ghz}, i.e., the expectation value $\langle \Pi_M^{(1)} \rangle_\rho$ with $\Pi_M^{(1)}$ to be the projector into the eigenspace of $M$ with eigenvalue $+1$. If one of $\meane{M}_\rho$ does not equal to $1$, we can conclude that this network state $\rho$ cannot be the state $|\mathrm{GHZ}\rangle\langle \mathrm{GHZ}|$.

Let us consider two kinds of inflation of the triangle-network as in Fig.~\ref{fig:GHZ}, where the corresponding network states are denoted as $\gamma, \eta$. Roughly speaking, the source between $B, C$ is broken in the $\gamma$ inflation. Note that inflation $\eta$ is actually a trivial inflation in triangular network.  By comparing Fig.~\ref{fig:GHZ}(a), Fig.~\ref{fig:GHZ}(b) and Fig.~\ref{fig:GHZ}(c), we have the marginal relations
\begin{align}\label{eq:ghzc1}
    \meane{Z_BZ_A^\dagger}_\gamma = \meane{Z_BZ_A^\dagger}_\rho&,\ 
    \meane{Z_AZ_C^\dagger}_\gamma = \meane{Z_AZ_C^\dagger}_\rho,\\
    \meane{Z_BZ_C^\dagger}_\eta &= \meane{Z_BZ_C^\dagger}_\gamma,\\ 
    \meane{(X_AX_BX_{C^\prime})^{ \floor{\frac{d}{2}} }}_\eta &= \meane{(X_AX_BX_C)^{ \floor{\frac{d}{2}} }}_\rho.
\end{align}
For convenience, here $\meane{Z_BZ_C^\dagger}_\eta$ stands for $\meane{Z_BZ_C^\dagger}_{\eta_{BC}}$, where $\eta_{BC}$ is the reduced state of $\eta$ on parties $B$ and $C$. We use such shorthand notations throughout the whole manuscript without confusion.
By applying Lemma~\ref{lm:product_ineq} in Appendix~A in Supplemental Information(SI),
we have
\begin{align}\label{eq:ghzc21}
    \meane{Z_BZ_C^\dagger}_\gamma \ge \meane{Z_BZ_A^\dagger}_\gamma + \meane{Z_AZ_C^\dagger}_\gamma - 1.
\end{align}
Under the assumption that $\meane{(X_AX_BX_{C'})^{\floor{\frac{d}{2}}}}_\eta \ge 1/2$ and $\meane{Z_BZ_C^\dagger}_\eta \ge 1/2$, Lemma~\ref{lm:fidelity},\ref{lm:eigenvalue},\ref{lm:incompatible} in SI, lead to
\begin{align}\label{eq:ghzc22}
    &|2\meane{(X_AX_BX_{C'})^{\floor{\frac{d}{2}}}}_\eta\!-\!1|^2\!+\! |2\meane{Z_BZ_C^\dagger}_\eta\!-\!1|^2\! \le\! 1\!+\!\sin\theta_d,
\end{align}
where $\theta_d=0$ when $d$ is even and $\theta_d=\frac{\pi}{2d}$ when $d$ is odd. This assumption can be ensured whenever $\mathcal{F}(|\mathrm{GHZ}\rangle\langle \mathrm{GHZ}|, \rho) \ge 3/4$ according to Lemma~\ref{lm:fidelity} in SI.

If $\rho$ is indeed the GHZ state, then Eq.~\eqref{eq:ghzc1}--~\eqref{eq:ghzc22} cannot hold simultaneously, which is a contradiction. Combined with Lemma~\ref{lm:fidelity}, which implies $\meane{(X_AX_BX_{C'})^{\floor{\frac{d}{2}}}}_\eta\ge \mathcal{F}$ and $\meane{Z_BZ_C^\dagger}_\eta\ge 2 \mathcal{F}-1$, Eq~\eqref{eq:ghzc22} can be solved as 
\begin{equation}\label{eq:boundghz}
    \frac{3}{4} \le \mathcal{F}(|\mathrm{GHZ}\rangle\langle \mathrm{GHZ}|, \rho) \le \frac{7+\sqrt{4+5\sin\theta_d}}{10}.
\end{equation}
Hence, either the assumption does not hold, then $\mathcal{F}(|\mathrm{GHZ}\rangle\langle \mathrm{GHZ}|, \rho) < 3/4$; or the assumption holds and so do the inequalities in Eq.~\eqref{eq:boundghz}. Since the upper bound in Eq.~\eqref{eq:boundghz} is always larger than $3/4$ for any dimension $d$, the upper bound in Eq.~\eqref{eq:boundghz} holds whatever the assumption holds or not. 
We remark that we only used GHZ states as an example to 
introduce our general method, a tighter bound in Appendix~D on the fidelity exists \cite{navascues2020genuine}.

\textit{Multi-qudit graph states.}---
The three-qudit GHZ state is a special case of a multi-qudit graph 
state~\cite{looi2008quantum}. In the same spirit, we can also derive no-go 
theorems for network states with multi-qudit graph states as targets. Each
graph state is associated with a  multigraph $G=(V,E)$, defined by its 
vertex set $V$, and the edge set $E$, where the edge between vertices 
$i,j$ with multiplicity $m_{ij}$ is denoted as $((i,j),m_{ij})$. Without loss of generality, we will only consider multi-partite graph states, i.e. number of parties (vertices) is at least 3. Due to the periodicity as follows, the multiplicity can be limited to $m_{ij}=0,1,\cdots,d-1$, where $d$ is dimension of Hilbert space for a single qudit.
Besides, we denote $N_i$ the neighborhood of vertex $i$, and define 
unitary stabilizers $g_i = X_i\otimes_{j\in N_i} Z^{m_{ij}}$. Note that $g_i^d=I$.
For a given multigraph $G=(V,E)$, the corresponding graph state is the unique common eigenvector of the operators $\{g_i\}_{i\in V}$ with eigenvalue $+1$ \cite{looi2008quantum,bahramgiri2006graph,PhysRevA.71.042315,schlingemann2001stabilizer,bahramgiri2006graph}.
This eigenvector $|G\rangle$ is called the graph state associated to the graph $G$.

\begin{theorem}
\label{ob: qudit constant multiplicity}
Multi-qudit graph states with connected graph and multiplicities either
being constant $m$ or $0$ cannot be prepared in any network with 
bipartite sources.
\end{theorem}

\begin{figure}[t]
    \centering
    \includegraphics[width=0.45\textwidth]{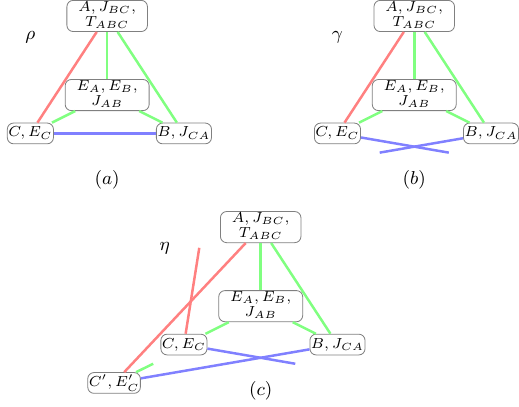}
    \caption{Fully connected network with bipartite sources and two kinds of its inflation, where the corresponding states are denoted 
    as $\rho, \gamma$ and $\eta$. The sources which have not been 
    changed in the whole proof are omitted in the figure.
    }
    \label{fig:network_full}
\end{figure}

\begin{proof}
Here we explain the main idea, the detailed proof is provided in Appendix~B in SI. 
The structure of the proof is in the same spirit as the one for $\ghz$ state.
The key point is to keep necessary marginal relations in different kinds of 
inflations and finally derive a contradiction. 

For a connected graph $G$ with no less than three vertices, there are 
always three vertices $A, B, C$ such that $(A,B), (A,C)$ are two edges. 
If $(B, C)$ is not an edge, we call $(A,B,C)$ an angle. Otherwise, 
we call $(A,B,C)$ a triangle.

To carry out the proof, we have to carefully group the neighborhoods of 
vertices $A, B, C$ and choose proper stabilizers of the graph state 
correspondingly. Here we consider the case where $(A,B,C)$ is a triangle 
as an example. Then we can partition all the vertices into $4$ groups 
as in Fig.~\ref{fig:network_full}, where $T_{ABC}$ is the common neighborhood 
of $A, B, C$, $J_{AB}$ is the common neighborhood of $A, B$ but not $C$, $E_A$ 
is the neighborhood which is not shared by $B, C$ and so on.

By choosing $S_1 = g_Ag_B^\dagger$, $S_2 = g_A^\dagger g_C$, $S_3 = S_1 S_2 = g_B^\dagger g_C$ and $S_4 = g_C^t$ (the $t$-th power of $g_C$), we have the 
marginal relations
\begin{align}\label{eq:gagb}
        &\meane{S_1}_\gamma = \meane{S_1}_\rho,\     \meane{S_2}_\gamma = \meane{S_2}_\rho,\\
        &\expct{S_3}_\eta = \expct{S_3}_\gamma,\ \expct{S_{4^\prime}}_\eta = \expct{S_4}_\rho,
\end{align}
where $S_{4^\prime}$ is related to $S_4$ by changing party $C^\prime$ in the support to $C$. 
These marginal relations can be verified by comparing the supports of each 
operator in different kinds of inflation. 

However, $S_3, S_{4^\prime}$ do not commute in the inflation $\eta$. More precisely,
we have that 
\begin{equation}\label{eq:angled}
    |\expct{S_3}_\eta|^2 + |\expct{S_{4^\prime}}_\eta|^2 \le 1 + \sin\theta_{t,d},
\end{equation}
where $0\le\theta_{t,d}\le\pi/6$ by choosing $t$ properly.

Similarly as the analysis for the GHZ state, the relation that $S_3 = S_1S_2$ and conditions in Eqs.~(\ref{eq:gagb} - \ref{eq:angled}) lead to 
\begin{equation}\label{eq:ghnetfidelity}
\mathcal{F}(|G\rangle\langle G|, \rho) \le \frac{7+\sqrt{4+5\sin(\theta_{t,d})}}{10}<0.95495.
\end{equation}
This is in fact a universal bound for arbitrary configuration of 
equal-multiplicity multi-qudit graph states.
\end{proof}

Note that for qubit graph states, the multiplicity is either $1$ or $0$, this leads to the following theorem.
\begin{theorem}\label{ob:qubit_fidelity}
Any multi-qubit graph state with connected graph cannot be prepared in a network with only 
bipartite sources, with $9/10$ as an upper bound of fidelity between 
graph state and network state. This follows from the fact that 
$\theta_{t,2} = 0$.
\end{theorem}

In order to formulate a more general statement, note that the key ingredients  in the proof of Theorem~\ref{ob: qudit constant multiplicity}, were that all parties can be grouped in a special 
way which fits to the algebraic relations $S_3 = S_1 S_2$ for
commuting $S_1, S_2$, moreover, it was needed that $S_3, S_{4'}$ 
have no common eigenvectors with eigenvalue $+1$. This leads to a
a more general theorem for the states with a set of stabilizers.

\begin{theorem}\label{ob:general}
For a given {pure} state $\sigma$ with {commuting} (unitary or projection) stabilizers $\{S_1, S_2, S_3=S_1 S_2, S_4\}$, it cannot be prepared in bipartite network if 
\begin{enumerate}
    \item all the parties can be grouped into $\{G_i\}_{i=1}^4$ such that $S_i$ has no support in $G_i$, see also Fig.~\ref{fig:general_operator}(a);
    \item {$S_1$, $S_2$ commutes, and} $S_3, S_{4'}$ have no common eigenvectors with eigenvalue $1$, where $S_{4'}$ has the support $G_2, G_3$ and a copy of $G_1$, and $S_{4'}$ acts there same as $S_4$ on $\supp{(S_4)}$.
\end{enumerate}
\end{theorem}

\begin{figure}[t]
\includegraphics[width=0.45\textwidth]{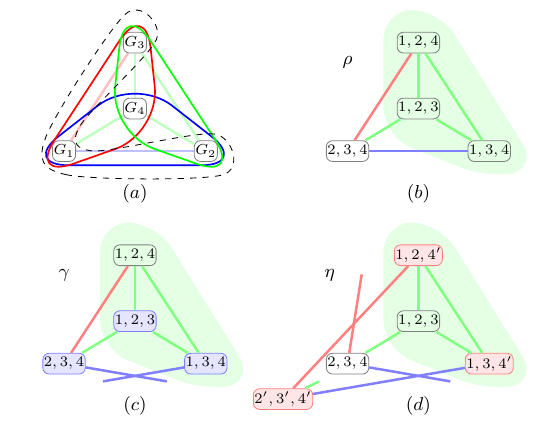}
    \caption{General inflation scheme in bipartite network similar to Fig.~\ref{fig:network_full}. In (a), $G_i$'s are group of parties in a partition, the green, red, blue and dashed circles stand for the supports for operators $S_1, S_2, S_3, S_4$, respectively. In (b), (c) and (d), we replace the label of each group by the indices of the $S_i$'s which has this group as support. The green shadow represents a multipartite source relating to all the groups in it. The sources which have not been changed in the whole proof are omitted in the figures.}
    \label{fig:general_operator}
\end{figure}

\begin{proof}
Here we provide the main steps of the proof without diving into details.
The first condition implies that, those four operators do not have common support for all of them. Hence, the set of parties can be divided into the four parts as illustrated in Fig.~\ref{fig:general_operator}(a).
By comparing the supports of the operators in different kinds of inflation as in Fig.~\ref{fig:general_operator}, we have the marginal relations
\begin{align}
    \meane{S_1}_\gamma = \meane{S_1}_\rho,~ \meane{S_2}_\gamma = \meane{S_2}_\rho,\\
    \meane{S_3}_\eta = \meane{S_3}_\gamma,~ \meane{S_{4'}}_\eta = \meane{S_4}_\rho.
\end{align}
Through those marginal relations, we can relate $\meane{S_3}_\eta$ and $\meane{S_{4'}}_\eta$ to $f = \mathcal{F}(\sigma,\rho)$. 
Finally, the second condition leads to the result that $f<1$. In fact, 
\begin{equation}\label{eq:fidelityupper}
  f \le \frac{7+\sqrt{4 + 5\lambda'/2}}{10} < 1,
\end{equation}
where  $\lambda' < 2$ is the maximal singular value of $\{S_3, S_{4'}^\dagger\}$.
\end{proof}

Note that all the bipartite sources in the network among parties 
in $\supp(S_1)$ have not been touched during the inflation procedure, 
so the proof still holds even if there is a multipartite source just 
affecting this set of parties. By exhaustive search and applying Theorem~\ref{ob:general} to multi-qudit graph states, we can 
figure out the situations where all $n$-partite qudit graph states 
in dimension $d$ cannot arise from a network with bipartite sources, 
see {Theorem~\ref{thm:qudit}, which is a consequence of Theorem~\ref{ob:general}, and} discussions therein in Appendix~C in the SI. The result is summarized in 
Table~\ref{table:graph state}. Moreover, any prime-dimensional 
graph state satisfying special structures as in Theorem~\ref{thm:qudit}
cannot be prepared in a network with bipartite sources.

\begin{table}[t!]
    \centering
    \begin{tabular}{|c|c|c|c|c|c|}
    \hline
         $d$ & 3 & 4 & 5, 7 & Prime dimension \\ \hline
         $n$ & 3,4,5,6 & 3,4,5 & 3,4 & 3\\ \hline
    \end{tabular}
    \caption{For a given dimension $d$ one can consider an $n$-partite graph state and ask whether Theorem~\ref{thm:qudit} can be used 
    to prove the impossibility of its preparation, see discussions in Appendix~C. This table identifies situations where {\it all} $n$-partite qudit graph states in dimension $d$ cannot be prepared by the network. 
    }
    \label{table:graph state}
\end{table}

\textit{Conclusion and Discussions.}---
Here we have developed a toolbox to compare multi-qudit graph states 
and states which are generated in a quantum network  without 
memory and feed forward. By combining those tools related to symmetry and the inflation technique, we proved that all multi-qudit graph states, where the non-zero multiplicities are a constant, cannot be prepared in the quantum network with just bipartite sources. 
The result can also be generalized to a larger class of multi-qudit graph states and quantum networks with multipartite sources, in the case that the generalization does not affect the necessary marginal relations during the inflation procedure. The more general case with multi-partite sources is an interesting topic for future research. Furthermore, we provided a fidelity estimation of the multi-qudit graph states and network states based on a simple analysis. 

More effort should be contributed to the fidelity analysis in the future, such as introducing more types of inflation in
Ref.~\cite{hansenne2022symmetries}, and generalizing the techniques 
in Ref.~\cite{navascues2020genuine} for the states other than GHZ 
states. 
Another interesting project for further study is to consider
other families of states, like Dicke states or multi-particle 
singlet states, which are not described by a stabilizer formalism. 
Finally, from the fact that graph states cannot be prepared in the 
simple model of a network considered here, the question arises, which
additional resources (such as classical communication) facilitate the generation of such states. 
Characterizing these resources will help 
to implement quantum communication in networks in the real world.

{\it Note added:} While finishing this manuscript, we became aware of a related work by O. Makuta {\it et al.}~\cite{makuta2022graph}.  Albeit those two works originate from the same conjecture in Ref.~\cite{hansenne2022symmetries}, the techniques and results are different from few perspectives. Especially, we have only made use of two kinds of inflation and Theorem~2 here holds for all dimensions but with limited multiplicities. The resulting fidelity between the graphs states and network states has also different estimations. The application of Theorem~3, like in combination with Theorem~10, can cover more situations.

\begin{acknowledgments}
We thank
{Carlos de Gois, David Gross, Kiara Hansenne, 
Laurens Ligthart, and Mariami Gachechiladze}
for discussions. This work was supported by the
Deutsche Forschungsgemeinschaft (DFG, German 
Research Foundation, project numbers 447948357 
and 440958198), the Sino-German Center for 
Research Promotion (Project M-0294), the ERC 
(Consolidator Grant No. 683107/TempoQ) and the 
German Ministry of Education and Research 
(Project QuKuK, BMBF Grant No. 16KIS1618K). 
Z.P.X. acknowledges support from the Alexander von 
Humboldt Foundation, {National Natural Science Foundation of China} (Grant No.\ 12305007) and 
Anhui Provincial Natural Science Foundation (Grant No.\ 2308085QA29).
\end{acknowledgments}

\section{DATA AVAILABILITY}

Data sharing not applicable to this article as no data sets were generated or analyzed
during the current study.

\section{CODE AVAILABILITY}

The code used to perform related numerical computation are available from the corresponding author upon reasonable request.

\section{AUTHOR CONTRIBUTIONS}

Y.X.W., Z.P.X. and O.G. derived the results and wrote the manuscript. O.G. supervised the project.

\section{COMPETING INTERESTS}

The authors declare no competing interests.

\newpage
\appendix
\onecolumngrid

\vspace{5em}
\begin{center}
    \textbf{
Supplementary Information of \\``Quantum LOSR networks cannot generate graph states with high fidelity"
}
\end{center}

\setcounter{equation}{15}
\addtocounter{theorem}{0}
\section{Appendix A. Analysis of operators}\label{app:operators}
For discussions about fidelity analysis, it is convenient to develop the following mathematical tools.

\begin{lemma}\label{lm:product_ineq}
For three unitary operators $S_1, S_2, S_3=S_1 S_2$ {which commute with each other}, we have
\begin{equation}
\meane{S_3}\geq \operatorname{Prob}(S_1=S_2=1)\geq \meane{S_1}+\meane{S_2}-1,
\end{equation}
where $\meane{M}_\rho$ is the expectation value $\langle \Pi_M^{(1)} \rangle_\rho$, and $\Pi_M^{(1)}$ is the projector into the eigenspace of $M$ with eigenvalue $1$.
\end{lemma}

\begin{proof}
Since $S_3 = S_1S_2$, thus $\Pi_{S_3}^{(1)} \succeq \Pi_{S_1}^{(1)} \Pi_{S_2}^{(1)}$, which implies
\begin{equation}
     \meane{S_3}=\operatorname{Prob}(S_3=1)\geq \operatorname{Prob}(S_1=S_2=1).
\end{equation}
On the other hand~\cite{navascues2020genuine}, 
\begin{align}
    \operatorname{Prob}(S_1=S_2=1) \geq \operatorname{Prob}(S_1=1)+\operatorname{Prob}(S_2=1)-1  = \meane{S_1}+\meane{S_2}-1.
\end{align}

\end{proof}

\begin{lemma}\label{lm:fidelity}
For a given state $\rho$ and a pure stabilizer state $\sigma$ with the unitary stabilizer group $\mathcal{S}$, $\forall S \in \mathcal{S}$, we have
\begin{equation}
   \meane{S}_\rho\ge\mathcal{F}(\sigma, \rho) ,\ |\left\langle S\right\rangle_\rho| \ge \left\langle \re S\right\rangle_\rho \ge 2\meane{S}_\rho - 1,  
\end{equation}
where $\mathcal{F}(\sigma, \rho)$ is the fidelity between $\sigma, \rho$, $\re S = (S + S^\dagger)/2$. 
\end{lemma}
\begin{proof}
$\meane{S}_\rho\ge\mathcal{F}(\sigma, \rho)$ holds since $\Pi_S^{(1)} \succeq \sigma$ and $\rho \succeq 0$.

For an arbitrary vector, we can decompose it as $|u\rangle + |v\rangle$ where $\langle u|v\rangle = 0$, $\Pi_S^{(1)} |v\rangle = |v\rangle$, $\Pi_S^{(1)} |u\rangle = 0$. 
Note that
\begin{align}
     \left({S+S^\dagger} - 4\Pi_S^{(1)} + 2\id\right) |v\rangle = 0.
\end{align}
This implies
\begin{align}
    (\langle u| + \langle v|) \left({S+S^\dagger} - 4\Pi_S^{(1)} + 2\id\right) (|u\rangle + |v\rangle)
    = \langle u|\left({S+S^\dagger}  + 2\id\right)|u\rangle
    \ge 0,
\end{align}
where the last inequality is from the fact that ${S+S^\dagger}$ is hermitian and all its eigenvalues are no less than $-2$. 

Thus,
\begin{align}
   2\langle \re  S\rangle_{\rho} = \left\langle S+S^\dagger\right\rangle_\rho \ge 4\meane{S}_\rho - 2 \ge 4\mathcal{F}(\sigma,\rho) - 2.
\end{align}
{On the other hand, the fact that $\langle \re  S\rangle_{\rho} = \re(\langle S\rangle_{\rho})$ is the real part of $\langle S\rangle_{\rho}$ implies that $|\left\langle S\right\rangle_\rho| \ge \left\langle \re S\right\rangle_\rho$.}
\end{proof}

\begin{lemma}\label{lm:eigenvalue}
For a given set of operators $\{S_i\}$,
\begin{equation}
    \sum_i |\expct{S_i}|^2 \le \frac{1}{2} \lambda_{\max}(\mathcal{S}),
\end{equation}
where $\mathcal{S}$ is a matrix with $(i,j)$-th element $\langle\{S_i, S_j^\dagger\}\rangle$, $\{ \cdot, \cdot \}$ is used as notation for the anticommutator, and $\lambda_{max}(S)$ is the maximal eigenvalue of $\mathcal{S}$.
\end{lemma}
\begin{proof}
For a given normalized complex vector $\{c_i\}$, denote $M=\sum_i c_i S_i$. We have
\begin{align}
    \langle \{M,M^\dagger\}\rangle = \sum_{i,j} c_i c_j^* \langle\{S_i, S_j^\dagger\}\rangle \le \lambda_{\max}(\mathcal{S}).
\end{align}
Note that $\mathcal{S}$ is hermitian.

Consequently, the inequality
$
\langle M \rangle \langle M^\dagger \rangle \leq \frac{1}{2}\langle \{M,M^\dagger\}\rangle
$
implies
\begin{equation}
\frac{1}{2}\lambda_{\max}(\mathcal{S}) \geq \langle M \rangle \langle M^\dagger \rangle = \left|\sum_{i} c_i  \expct{S_i} \right|^2.
\end{equation}
Since $\{c_i\}$ can be an arbitrary normalized vector, we have
\begin{equation}
\frac{1}{2}\lambda_{\max}(\mathcal{S}) \geq \sum_{i} |\expct{S_i}|^2.
\end{equation}
\end{proof}
In the case that $\mathcal{F}(\sigma, \rho) \ge 1/2$, by combining Lemma~\ref{lm:fidelity} and Lemma~\ref{lm:eigenvalue}, we have
\begin{equation}\label{eq:eigenvalue2}
    \frac{1}{2}\lambda_{\max}(\mathcal{S}) \geq \sum\nolimits_i |2\meane{S_i}_{\rho} -1|^2 ,
\end{equation}

\begin{lemma}\label{lm:incompatible}
For a given set of unitary operators $\{S_i\}_{i=1}^n$ satisfying $S_iS_j = -e^{i\theta_{ij}} S_jS_i$ where $\theta_{ij} \in (-\pi,\pi]$, we have
\begin{equation}\label{eq:incompatibility}
    \lambda_{\max} (\mathcal{S}) \le 2[1 + (n-1)\sin (\theta/2)],
\end{equation}
where $\rho$ is an arbitrary mixed state, $\theta = \max_{i\neq j} |\theta_{ij}|$.
\end{lemma}
\begin{proof}
In the case that $\{S_i\}$ is a set of unitaries, and $S_iS_j = -e^{i\theta_{ij}} S_j S_i$, we have 
\begin{equation}
     S_j^\dagger S_i=S_j^\dagger S_i S_j S_j^\dagger = -e^{i\theta_{ij}} S_j^\dagger S_j S_iS_j^\dagger=-e^{i\theta_{ij}}S_iS_j^\dagger.
\end{equation}
Consequently,
\begin{align}\label{eq:element}
    |\langle\{S_i, S_j^\dagger\}\rangle| = |(1-e^{i\theta_{ij}}) \langle S_i S_j^\dagger \rangle| \le 2 |\sin (\theta_{ij}/2)|.
\end{align}
According to Frobenius theorem, 
\begin{equation}
\lambda_{\max} (\mathcal{S}) \le \max_i \sum_{j} |\langle\{S_i, S_j^\dagger\}\rangle|.
\end{equation}
Together with Eq.~\eqref{eq:element}, this leads to
\begin{equation}
\lambda_{\max} (\mathcal{S}) \le 2[1 + (n-1) \sin(\theta/2)].
\end{equation}
\end{proof}

We remark that, in the case that $S_iS_j = -S_jS_i$, i.e., $\theta_{ij}=0$, Eq.~\eqref{eq:incompatibility} reduces to the result proposed in Ref.~\cite{toth2005entanglement}, which have been used in the comparison of qubit graph states and network states in Ref.~\cite{hansenne2022symmetries}. However note that these inequalities are not tight. If we choose an eigenstate of $S_1$ and suppose $\theta_{1i}\neq \pi,\ \forall i\neq 0 $, then $\langle S_i\rangle=\langle S_1 S_i\rangle = -e^{i\theta_{1i}}\langle S_i S_1\rangle=-e^{i\theta_{1i}}\langle S_i\rangle,\  \forall i\neq 1$, which implies $\langle S_i\rangle=0,\ \forall i\neq 1$, which does not saturate the upper bounds unless all $\theta_{ij}=0$.

\section{Appendix B. Proof of Theorem 1 and Theorem 3}\label{app:Proof of Theorem 1 and Theorem 3}
Here we provide the rest of proof of Theorem~\ref{ob: qudit constant multiplicity} and the analysis of the fidelity.
To proceed, we introduce some necessary notations.
\begin{definition}
Given a network $N = (V, E)$ and a subset $T$ of vertices, the reduced network {is given by} $N|_T = (T, E|_T)$ where
$E|_T = \{e\cap T| e\in E\}$  and $e\cap T$ is the intersection of $T$ and $e$ by treating $e$ as a subset of vertices.
\end{definition}
 We have one remark. For two edges $e$ and $e'$ in different kinds of inflation, we say $e\cap T = e'\cap T$ if and only if $e\cap T$ and $e'\cap T$ have same elements, besides, $e$ and $e'$ are copies from the same edge in the original network.
\begin{lemma}\label{lm:inflation}
Consider two given different kinds of inflation $N_\gamma, N_\eta$ of a network, then the reduced states $\tr_{{T}_c}(\gamma) = \tr_{{T}_c}(\eta)$ if the reduced networks obey $N_\gamma|_T = N_\eta|_T$, where ${T}_c$ is the complement set of $T$.
\end{lemma}

\setcounter{theorem}{0}
\begin{theorem}
\label{ob: qudit constant multiplicity2}
Qudit graph states with connected graph and multiplicities either the constant $m$ or $0$ cannot be prepared in any network with bipartite sources.
\end{theorem}
\addtocounter{theorem}{6}
\begin{figure}[t]
    \centering
    \includegraphics[width=0.45\textwidth]{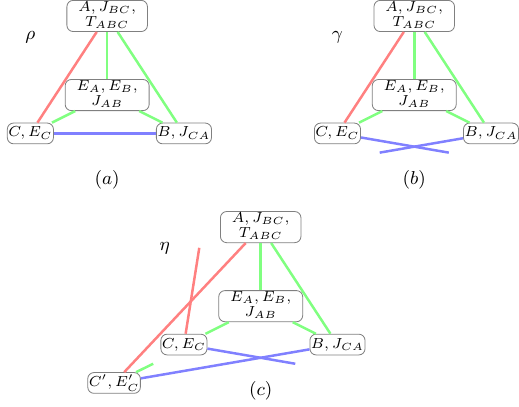}
    \caption{{\textcolor{black}{Fully connected network with bipartite sources and two kinds of its inflation, where the corresponding states are denoted 
    as $\rho, \gamma$ and $\eta$. The sources which have not been 
    changed in the whole proof are omitted in the figure.}}
    }
    \label{fig:network_full_app}
\end{figure}
\begin{proof}
{\textcolor{black}{
We first prove the theorem for the case where $(A,B,C)$ is a triangle, then prove for the angle case.\\
For convenience, we denote $T_{ABC}$ the common neighborhood of $A, B, C$ in the graph $G$, $J_{AB}$ the common neighborhood of $A, B$ but not $C$, $E_A$ the neighborhood which is not shared by $B, C$ and so on. Those sets of parties can be grouped as in Fig.~\ref{fig:network_full_app}. The structure of the proof is in the same spirit as the one for $\ghz$ state.
The key point is to keep necessary marginal relations in different kinds of inflation and derive contradiction. \\
Let us choose the three operators $S_1 = g_Ag_B^\dagger$, $S_2 = g_A^\dagger g_C$, $S_3 = S_1 S_2 = g_B^\dagger g_C$. By definition,
\begin{equation}
    S_1 \! =\! (X_AZ_A^{-m})(X_B Z_B^{-m})^\dagger Z^m_{E_A} Z^m_{J_{CA}} (Z^m_{E_B} Z^m_{J_{BC}})^\dagger
\end{equation}
has support $\supp(S_1) = \{A, B, E_A, E_B, J_{CA}, J_{BC}\}$.
Since $N_\rho|_{\supp(S_1)} = N_\gamma|_{\supp(S_1)}$, we have
\begin{align}\label{eq:gagb_app}
    \meane{S_1}_\gamma = \meane{S_1}_\rho.
\end{align}
Similarly, $\supp(S_2) = \{A, C, E_A, E_C, J_{AB}, J_{BC}\}$ and $N_\rho|_{\supp (S_2)} = N_\gamma|_{\supp(S_2)}$ lead to
\begin{align}\label{eq:gagc}
    \meane{S_2}_\gamma = \meane{S_2}_\rho.
\end{align}
The fact that $\supp (S_3) = \{B, C, E_B, E_C, J_{AB}, J_{CA}\}$ and $N_\eta|_{\supp (S_3)} = N_\gamma|_{\supp (S_3)}$ leads to
\begin{align}\label{eq:gbgc}
    \expct{S_3}_\eta = \expct{S_3}_\gamma.
\end{align}
Note that $\supp (g_C) = \{C, A, B, E_C, J_{CA}, J_{BC}, T_{ABC}\}$, $N_\rho|_{\supp (g_C)}$ is isomorphic to $N_\eta|_{\supp (g_{C^\prime})}$ with $\supp (g_{C^\prime}) = \{C^\prime, A, B, E_{C}^\prime, J_{CA}, J_{BC}, T_{ABC}\}$ by relabeling $C$ and $E_C$ by $C^\prime$ and $E_C^\prime$ in all parts. \\
If we denote by $S_4 = g_C^t$ the $t$-th power of $g_C$, this implies 
that
\begin{equation}\label{eq:gc}
    \expct{S_4}_\rho = \expct{S_{4^\prime}}_\eta,
\end{equation}
where $S_{4^\prime}$ is related to $S_4$ by changing party $C^\prime$ in the support to $C$. \\
If we assume the graph state $|G\rangle$ can be generated as the network state $\rho$, Eqs.~(\ref{eq:gagb_app}-\ref{eq:gc}) lead to $\expct{S_{4^\prime }}_\eta = \expct{S_3}_\eta = 1$, which conflicts with the fact that
\begin{equation}\label{eq:angled_app}
    |\expct{S_3}_\eta|^2 + |\expct{S_{4^\prime}}_\eta|^2 \le 1 + \sin\theta_{t,d},
\end{equation}
where $\sin\theta_{t,d}=|\cos(tm\pi/d)|$, and we can choose appropriate $t$ such that $0\le\theta_{t,d}\le\pi/6$, see below.
In fact, similar to the analysis for the GHZ state, conditions in Eqs.~(\ref{eq:gagb_app}-\ref{eq:angled_app}) lead to that 
\begin{equation}\label{eq:ghnetfidelity_app}
\mathcal{F}(|G\rangle\langle G|, \rho) \le \frac{7+\sqrt{4+5\sin(\theta_{t,d})}}{10}<0.95495.
\end{equation}
This is a universal bound for arbitrary configuration of 
equal-multiplicity multi-qudit graph states.
}}

Now we prove for the angle case, with same spirit as the triangle case. It suffices to explain how the two conditions in Theorem~\ref{ob:general_app} can be satisfied. 

Without loss of generality, we assume $(A, B, C)$ is the angle in the graph for the graph state, that is, there are only two edges $(A,B)$ and $(A,C)$.
Now we choose $S_1= g_C^\dagger$, $S_2=g_B$ and $S_4=g_A^t$, then  $S_3=g_B g_C^\dagger$.
The supports of the operators are
\begin{align}
    &\supp(S_1) = \{C, A, E_C, J_{CA}, J_{BC}, T_{ABC}\},\\
    &\supp(S_2) = \{B, A, E_B, J_{AB}, J_{BC}, T_{ABC}\},\\
    &\supp(S_3) = \{C, B, E_B, E_C, J_{AB}, J_{CA}\},\\
    &\supp(S_4) \subseteq \{A, B, C, E_A, J_{AB}, J_{CA}, T_{ABC}\}.
\end{align}

{
By grouping the parties into four groups as illustrated in Fig.~\ref{fig:triangle2}, that is,
\begin{align}
    G_1= \{B, J_{AB}\},\ 
    G_2 = \{C, J_{CA}\},\ 
    G_3 = \{A,E_A, T_{ABC}\},\ 
    G_4= \{E_B,E_C, J_{BC}\},
\end{align} 
One can verify that $\supp(S_i) \cap G_i = \emptyset$, for $i=1,2,3,4$. Hence, the first condition in Theorem~\ref{ob:general} is satisfied.
}
{
Since the stabilizers $S_i$'s are in product form, to verify the second condition in Theorem~\ref{ob:general} is equivalent to verify that $S_3|_{G_2}$ and $S_{4'}|_{G_2}$ has no common eigenvector with eigenvalue $1$.
}

\begin{figure}[t]
    \centering
\includegraphics[width=0.8\textwidth]{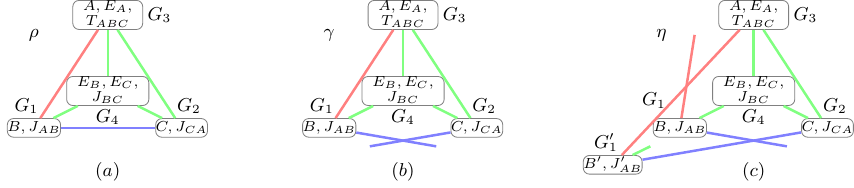}
    \caption{Three types of inflation and support of related regions in the angle case.}
    \label{fig:triangle2}
\end{figure}

As we can see, 
\begin{equation}
    S_3|_{G_2} = X_C^\dagger Z_{J_{CA}}^m,\ S_{4'}|_{G_2} = Z_C^{tm} Z_{J_{CA}}^{tm},
\end{equation}
this implies that $S_3|_{G_2}$ and $S_{4'}|_{G_2}$ cannot share a common eigenstate because $Z_C^{tm}X_C^\dagger = \omega^{tm} X_C^\dagger Z_C^{tm}$ where $t$ is chosen such that $\omega^{tm} \neq 1$ rather closest to $-1$. With the inflation process as shown in Fig.~\ref{fig:triangle2}, we also conclude that the graph state with an angle cannot originate from bipartite network.

We continue to estimate the fidelity $\mathcal{F}(\sigma,\rho)$ between the graph state $\sigma$ and the network state $\rho$ in both cases with either angle or triangle.
In the case that $S_3 S_{4'} = \omega^{tm} S_{4'} S_3$ with $\omega = \exp{2\pi i/d}$, 
\begin{equation}
    \{S_3, S_{4'}^\dagger\} = S_3S_{4'}^\dagger + S_{4'}^\dagger S_3 = (1+\omega^{tm}) S_3S_{4'}^\dagger,
\end{equation}
which implies that the maximal absolute value of the eigenvalues of $\{S_3, S_{4'}^\dagger\}$, i.e., $\lambda'$ is no more than $|1+\omega^{tm}| =  2|\cos(tm\pi /d)|${\textcolor{black}, thus justify the value of $\theta_{t,d}$ above}. Here we claim that $|\cos(tm\pi /d)| \le 1/2$ by choosing $t$ properly. By inserting this upper bound into Eq.~\eqref{eq:fidelityupper}, we know that 
\begin{equation}
    \mathcal{F}(\sigma, \rho) \le \frac{7+\sqrt{4 + 5/2}}{10} \approx 0.95495.
\end{equation}

Notice that, $m \in \{1,\ldots, d-1\}$. Denote $m' = m/\operatorname{gcd}(m,d), d'= d/\operatorname{gcd}(m,d)$ where $\operatorname{gcd}(m,d)$ is the greatest common divisor of $m,d$. Thus, $m', d'$ are relatively prime and $d'\ge 2$. Consequently, there exists integers $t_0$ and $k$ such that $t_0m' = kd' + 1$. If we choose $t = t_0 \lfloor d'/2 \rfloor$, then 
\begin{equation}
    \frac{tm'}{d'} = \frac{ (kd'+1)\lfloor d'/2 \rfloor}{d'} = k \lfloor d'/2 \rfloor + \frac{ \lfloor d'/2 \rfloor}{d'}.
\end{equation}
It is direct to see that $-1/6 \le { \lfloor d'/2 \rfloor}/{d'} - 1/2 \le 0$.
Hence,
\begin{equation}
    \sin\theta_{t,d}\left|\cos\left(\frac{tm\pi}{ d}\right)\right| = \left|\cos\left(\frac{ \lfloor d'/2 \rfloor \pi}{d'}\right)\right| =
    \begin{cases}
    0, & d' \text{ is even},\\
    \sin\left(\pi/(2d')\right), & d' \text{ is odd},
    \end{cases}
\end{equation}
which implies $|\cos(tm\pi/d)| < 1/2$ directly.

\end{proof}

\begin{figure}[t]
\includegraphics[width=0.45\textwidth]{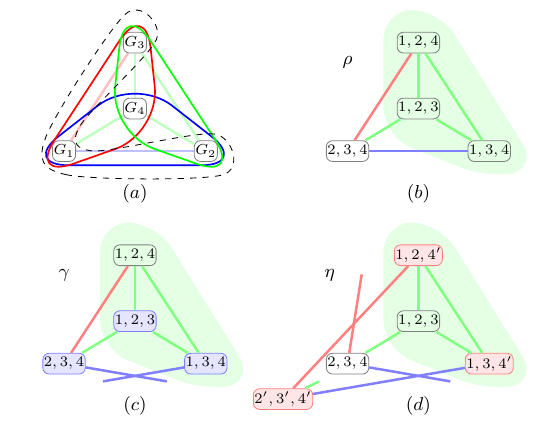}
    \caption{General inflation scheme in bipartite network similar to Fig.~\ref{fig:network_full}. In (a), $G_i$'s are group of parties in a partition, the green, red, blue and dashed circles stand for the supports for operators $S_1, S_2, S_3, S_4$, respectively. In (b), (c) and (d), we replace the label of each group by the indices of the $S_i$'s which has this group as support. The green shadow represents a multipartite source relating to all the groups in it. The sources which have not been changed in the whole proof are omitted in the figures.}
    \label{fig:general_operator_app}
\end{figure}
Here we also present proof of Theorem~\ref{ob:general_app}.
\addtocounter{theorem}{-5}
\begin{theorem}\label{ob:general_app}
For a given {pure} state $\sigma$ with {commuting} (unitary or projection) stabilizers $\{S_1, S_2, S_3=S_1 S_2, S_4\}$, it cannot be prepared in bipartite network if 
\begin{enumerate}
    \item all the parties can be grouped into  $\{G_i\}_{i=1}^4$ such that $S_i$ has no support in $G_i$, see also Fig.~\ref{fig:general_operator}(a);
    \item {$S_1$, $S_2$ commutes, and} $S_3, S_{4'}$ have no common eigenvectors with eigenvalue $1$, where $S_{4'}$ has the support $G_2, G_3$ and a copy of $G_1$, and $S_{4'}$ acts there same as $S_4$ on $\supp{(S_4)}$.
\end{enumerate}
\end{theorem}
\addtocounter{theorem}{6}
\begin{proof}
Here we provide the main steps of the proof without diving into details.
The first condition implies that, those four operators do not have common support for all of them. Hence, the set of parties can be divided into the four parts as illustrated in Fig.~\ref{fig:general_operator_app}(a).
By comparing the supports of operators $S_1, S_2$ in the original network in Fig.~\ref{fig:general_operator_app}(b) and in the inflation in Fig.~\ref{fig:general_operator_app}(c), we have the marginal relations
\begin{equation}
    \meane{S_1}_\gamma = \meane{S_1}_\rho,~ \meane{S_2}_\gamma = \meane{S_2}_\rho.
\end{equation}
The fact that $S_3 = S_1S_2$ implies
\begin{equation}
    \meane{S_3}_\gamma \ge \meane{S_1}_\gamma + \meane{S_2}_\gamma - 1.
\end{equation}
From the comparison of support of operator $S_3$ in the inflation related to $\gamma$ and the inflation related to $\eta$ as shown in Fig.~\ref{fig:general_operator_app}(c) and Fig.~\ref{fig:general_operator_app}(d), we know that
\begin{equation}
    \meane{S_3}_\eta = \meane{S_3}_\gamma.
\end{equation}
Similarly, by comparing the original network and the one related to $\eta$, it holds that
\begin{equation}
    \meane{S_{4'}}_\eta = \meane{S_4}_\rho.
\end{equation}
Since $S_3, S_{4'}$ have no common eigenvectors with eigenvalue $1$, the values $\meane{S_3}_\eta$ and $\meane{S_{4'}}_\eta$ cannot be $1$ at the same time, which contradicts with the assumption that $S_1, S_2, S_4$ are stabilizers of $\rho$.
\end{proof}

 Here let us also study bounds on fidelity. Denote $\lambda'$ the maximal absolute value of the eigenvalues of $\{S_3, S_{4'}^\dagger\}$.
{If we assume that $f = \mathcal{F}(\sigma,\rho) \ge 3/4$,}
Lemma~\ref{lm:eigenvalue} and Eq.~\eqref{eq:eigenvalue2} imply that
\begin{align}
    1+\lambda'/2 &\ge 1+ |\langle\{S_3, S_{4'} \}\rangle_{\eta}|/2\nonumber\\
    &\geq (2\meane{S_3}_\eta -1)^2 +(2\meane{S_{4'}}_\eta -1)^2.
\end{align}
Besides, by using Lemma~\ref{lm:product_ineq} and Lemma~\ref{lm:fidelity}, we have
\begin{equation}
    \meane{S_3}_\eta \ge 2f -1,\ \meane{S_{4'}}_\eta \ge f.
\end{equation}
The combination of them leads to
\begin{equation}\label{eq:fidelityupper_app}
  \frac{7 - \sqrt{4 + 5\lambda'/2}}{10} \le  f \le \frac{7+\sqrt{4 + 5\lambda'/2}}{10},
\end{equation}
where the lower bound is smaller than $3/4$ and the upper bound is larger than $3/4$. Hence, whatever the assumption $f\ge 3/4$ holds or not, the upper bound in Eq.~\eqref{eq:fidelityupper} always holds.  

\section{Appendix C. Another theorem of qudit graph states}\label{app:Another theorem of qudit graph states}

We have discussed qudit graph states with constant multiplicities in the main text. Theorem~\ref{ob:general} can also tackle the cases with different multiplicities.
Denote $m_{ij}$ the multiplicity of the edge $(i,j)$ in the graph for the qudit graph state hereafter.
\begin{theorem}
\label{thm:qudit}
For a graph states in arbitrary dimension, if  the graph satisfies one of the following conditions, this graph states cannot be prepared in arbitrary bipartite network. These conditions are
\begin{enumerate}
    \item There is a triangle $(A, B, C)$ in the graph, and  $J_{XY}=J_{YZ}=T_{ABC}=\emptyset$, and $(m_{XY} m_{YZ}/h) \pmod{d} \neq 0$, where $h = \gcd(m_{AB}, m_{CA}, m_{BC})$ and $X, Y, Z$ is a permutation of $A, B, C$.
    \item There is an angle $(A, B, C)$ in the graph, where $(B,C)$ is not an edge, $T_{ABC}=\emptyset$, 
    and $(m_{AB} m_{CA}/h') \pmod{d} \neq 0$, where $h' = \gcd(m_{AB}, m_{CA})$.
\end{enumerate}
We also have an upper bound for fidelity less than $0.95495$, but it differs for different dimensions and multiplicities, so we will leave discussions about fidelities for later.
\end{theorem}

{We remark that, the condition about multiplicities is always satisfied when $d$ is a prime number, i.e., this theorem works for arbitrary non-zero multiplicities $m_{XY}, m_{YZ}$ for prime-dimensional graph states.}

\begin{proof}
Here we explain how the two conditions in Theorem~\ref{ob:general} can be satisfied. 

\textit{1. The triangle case}

Assume $(A, B, C)$ is the triangle in the graph. Without loss of generality, we assume that $Y = A, X = B, Z = C$. We choose stabilizers $S_1=g_A^a g_C^c$, $S_2=g_A^{-a} g_B^{-b}$, $S_3=S_1 S_2= g_B^{-b} g_C^{c}$ and $S_4=g_A^t$ with the parameters $(a,b,c,t)$. To carry out the same inflation procedure as in the proof of Theorem~\ref{ob:general}, one option is to have $S_1|_B = \id_B, S_2|_C = \id_C, S_3|_A = \id_A$. For those conditions to be satisfied, we need
\begin{align}
    d| (a\, m_{AB} + c\, m_{BC}), \quad d| (-a\, m_{CA} - b\, m_{BC}),
    \quad d| (-b\,  m_{AB} + c\, m_{CA}),
\end{align}
which are satisfied by setting
\begin{equation}
    c= m_{AB}/h,\ b= m_{CA}/h,\ a = -m_{BC}/h.
\end{equation}
Under the assumption that $J_{AB}=J_{CA}=T_{ABC}=\emptyset$, we have
\begin{align}
    &\{C, A\} \subseteq \supp(S_1) \subseteq \{C, A, E_C, E_A, J_{BC}\},\\ 
    &\{A, B\} \subseteq \supp(S_2) \subseteq \{A, B, E_A, E_B, J_{BC}\},\\
    &\{B, C\} \subseteq \supp(S_3) \subseteq \{B, C, E_B, E_C, J_{BC}\},\\ 
    &\supp(S_4) \subseteq \{A, B, C, E_A\}.
\end{align}
{
By grouping the parties into four groups similar to Fig.~\ref{fig:network_full} {but with $J_{AB}=J_{CA}=T_{ABC}=\emptyset$}, that is,
\begin{align}
    G_1= \{B\},\ G_2 = \{C\},\ G_3 = \{A,E_A\},\ G_4= \{E_B,E_C, J_{BC}\},
\end{align}
one can verify that $\supp(S_i) \cap G_i = \emptyset$, for $i=1,2,3,4$. Hence, the first condition in Theorem~\ref{ob:general} is satisfied.
}
{
Similar to Theorem~\ref{ob: qudit constant multiplicity}, it suffices to verify that $S_3|_{G_2}$ and $S_{4'}|_{G_2}$ has no common eigenvector with eigenvalue $1$.
}
{
From $S_{4'}|_{G_2} = Z^{tm_{CA}}_C , S_3|_{G_2} = Z^{-b m_{BC}}X^c$, we have $S_3S_{4'} = \omega^{-t \tilde{m}} S_{4'}S_3$, where $\tilde{m} = m_{AB}m_{CA}/h \pmod{d}$. By assumption, $1\le \tilde{m} \le d-1$.
}
By choosing $t$ suitably as discussed in Appendix~B, the second condition can be satisfied. We also have the same upper bound of fidelity between the graph state and the network state, i.e., $0.95495$.

\textit{2. The angle case}

Denote $(A,B,C)$ the angle, equivalently, $m_{AB},m_{CA}\neq 0$ but $m_{BC}=0$. In this case, we choose $S_1=g_C^c$, $S_2= g_B^b$,  $S_3=S_1S_2 = g_C^c g_B^b$ and $S_4 = g_A^{t}$. 
Automatically, $S_1|_B = \id_B, S_2|_C = \id_C$. As for $S_3|_A = \id_A$, this requires $d| (b\, m_{AB}+ c\, m_{CA})$, which is satisfied with
\begin{equation}
    c = m_{AB}/h',\ b = -m_{CA}/h'.
\end{equation}
Under the assumption that $T_{ABC}=\emptyset$ and and $(m_{AB} m_{CA}/h') \pmod{d} \neq 0$, we have
\begin{align}
    &\{C, A\} \subseteq \supp(S_1) \subseteq \{C, A, E_C, J_{CA}, J_{BC}\},\\ 
    &\{B,A\} \subseteq \supp(S_2) \subseteq \{B, A, E_B, J_{AB}, J_{BC}\},\\
    &\supp(S_3) \subseteq (\supp(S_1) \cup \supp(S_2) )\setminus \{A\},\\ 
    &\supp(S_4) \subseteq \{A, B, C, E_A, J_{AB}, J_{CA}\}.
\end{align}
{
By grouping the parties into four groups {similar to Fig.~\ref{fig:triangle2} but with $T_{ABC}=\emptyset$} , that is,
\begin{align}
    G_1= \{B, J_{AB}\},\ G_2 = \{C, J_{CA}\},\ G_3 = \{A,E_A\},\ G_4= \{E_B,E_C, J_{BC}\}.
\end{align}
One can verify that $\supp(S_i) \cap G_i = \emptyset$, for $i=1,2,3,4$. Hence, the first condition in Theorem~\ref{ob:general} is satisfied.
}
{
Similar to Theorem~\ref{ob: qudit constant multiplicity}, it suffices to verify that $S_3|_{G_2}$ and $S_{4'}|_{G_2}$ has no common eigenvector with eigenvalue $1$.
}
The fact that
\begin{equation}
    S_{4^\prime}|_{G_2}=Z_C^{t m_{CA}} \prod_{i\in J_{CA}} Z_i^{t m_{iA}},\ S_{3}|_{G_2}=X_C^{c} \prod_{i\in J_{CA}} Z_i^{c m_{iC}}
\end{equation}
leads to $S_3 S_{4^\prime}= \omega^{-ctm_{CA}} S_{4^\prime} S_3 = \omega^{-t \tilde{m}} S_{4^\prime} S_3 $, where $\tilde{m} = m_{AB} m_{CA}/h' \pmod{d}$. 
By assumption, $1\le \tilde{m} \le d-1$.
By choosing $e$ suitably as discussed in Appendix~B, the second condition can be satisfied. Further, we result in the same upper bound of fidelity between the graph state and the network state, i.e., $0.95495$.
\end{proof}

In practice, the conditions in Theorem~\ref{thm:qudit} are not applicable in some cases, for example, the condition $T_{ABC}=\emptyset$ does not hold for the fully connected graph states. However, one can still try to find out other kinds of stabilizers which fulfill the two conditions in Theorem~\ref{ob:general}.
Another option is to convert the graph state to another one, which satisfies the conditions in Theorem~\ref{thm:qudit}, by applying local complementation operation\cite{LCPhysRevA.82.062315}. In the case of qudit graph states,  this operation on a given vertex $a$ of the graph is
\begin{equation}
m_{ij} \rightarrow (m_{ij}+m_{ai}\times m_{aj}) \mod d,\quad \forall i,j\in N_a,
\end{equation}
where $N_a$ is the neighborhood of vertex $a$.
For the qubit graph state, any non-zero $m_{ij}$ can only be $1$. Consequently, the local complementation operation on a given vertex $a$ is $m_{ij} \to 1-m_{ij}, \forall i,j \in N_a$.

It is easy to see that the angle $(A, B, C)$ can be mapped to the triangle by local complementation for qubit graph state. This is not true generally for the qudit graph states.

Now we can determine whether a given graph state with connected graph  can be prepared in a bipartite quantum network based on Theorem~\ref{thm:qudit}. For a given graph, we can run over all possible local complementation operations and determine if there exists any equivalent graph such that one of the two conditions is satisfied. If so, then this graph (and the equivalent graphs) cannot be prepared in any bipartite network. Running over $d=3$, $n=3,4,5,6$; $d=4$, $n=3,4,5$; and $d=5,7$, $n=3,4$ cases, where $d$ is local dimension and $n$ is number of parties, we found all of these graph states cannot be prepared in bipartite network. And for prime dimension and $n=3$, the conditions always hold. On the one hand, conditions on support is automatically true for $n=3$. On the other, since $d$ is prime, $(m_{XY} m_{YZ}/h) \pmod{d} \neq 0$ and $(m_{AB} m_{CA}/h') \pmod{d} \neq 0$ also always hold. Thus, those graph states cannot be prepared in any bipartite network. We then summarize those results in Table~\ref{table:graph state} in the main text.

We remark that, even up to local complementation, the two conditions in Theorem~\ref{ob:general} cannot be satisfied for the $6$-dimensional tripartite qudit angle graph states, where $m_{AB}=3$, $m_{CA}=2$ and $m_{BC}=0$.

\section{Appendix D. Better bounds for fidelity of tripartite GHZ states}\label{app:Better bounds for fidelity of GHZ states}

Now we will give a more detailed study on the bounds for fidelity of tripartite $d$-dimensional GHZ states. For tripartite GHZ states we can impose additional constraints, which results in tighter bound.

\subsection{1. Two extra lemmas}
Firstly we introduce extra inequalities which can improve the estimation of fidelity a little bit.
\begin{lemma}\label{lm:corr_sum}
Consider two commuting unitaries $S_1, S_2$, denote $S_3 = S_1 S_2$, then
\begin{align}
    &\langle \re  S_3 \rangle \geq \langle \re  S_1\rangle+\langle \re  S_2\rangle -3/2 ,\label{eq:meaninequality}\\
    &|\langle S_3\rangle|^2>|\langle S_1\rangle|^2+|\langle S_2\rangle|^2 -3/2.\label{eq:meansuare}
\end{align}
\end{lemma}

\begin{proof}
Since $S_1, S_2$ commute with each other, without loss of generality, we can assume they are diagonalized. Denote $e^{i \theta_{1j}}$ the $j$-th diagonal item of $S_1$, 
$e^{i \theta_{2j}}$ the $j$-th diagonal item of $S_2$.

Then the $j$-th diagonal item of $(S_3 - S_1 - S_2) + (S_3 - S_1 - S_2)^\dagger$ is $2(\cos(\theta_{1j} + \theta_{2j}) - \cos(\theta_{1j}) - \cos(\theta_{2j}))$, which is no less than $-3$ as we can verify with direct optimization. With this, we conclude the inequality in Eq.~\eqref{eq:meaninequality}.

Denote $p_j$ the $j$-th diagonal element of the state $\rho$ for the mean value, we have
\begin{align}
|\langle S_3\rangle |^2- |\langle S_1\rangle|^2 -|\langle S_2\rangle |^2 &= \left|\sum\nolimits_j p_j  e^{i\theta_{1j}+i\theta_{2j}}\right|^2-\left|\sum\nolimits_j p_j  e^{i\theta_{1j}}\right|^2-\left|\sum\nolimits_j p_j  e^{i\theta_{2j}}\right|^2 \nonumber\\
&=\sum\nolimits_{i,j} p_i p_j [\cos(\tilde{\theta}_{1ij}+\tilde{\theta}_{2ij}) -\cos(\tilde{\theta}_{1ij})-\cos(\tilde{\theta}_{2ij})]\nonumber\\
&\ge -\sum\nolimits_j p_j^2 -3\sum\nolimits_{i<j} p_i p_j\nonumber\\
&=-\left[3\left(\sum\nolimits_j p_j\right)^2 - \sum\nolimits_j p_j^2\right]/2\nonumber\\
&\ge -(3-1/d)/2,
\end{align}
where $\tilde{\theta}_{1ij} = \theta_{1i}-\theta_{1j}$, $\tilde{\theta}_{2ij} = \theta_{2i}-\theta_{2j}$ and $d$ is dimension. With this, we prove the inequality in Eq.~\eqref{eq:meansuare}.

\end{proof}

\begin{lemma}\label{lm:uncertainty}
For two unitary operator $S_1 S_2 = - e^{i\theta} S_2 S_1$, where $\theta\in (-\pi, \pi]$, 
then we have
\begin{equation}
\label{eq:uncertainty}
    \var(\re S_1)\var(\re S_2) \geq [R\left( \langle \re S_1\rangle \langle \re S_2\rangle-|\sin(\theta/2)|\right)]^2,
\end{equation}
where $R(x) = \max \{0, x\}$.
\end{lemma}

\begin{proof}
Note $S_1 S_2 = -e^{i\theta} S_2 S_1$ implies $S_1 S_2^\dagger= -e^{-i\theta} S_2^\dagger S_1 $, $S_1^\dagger S_2 = -e^{-i\theta} S_2 S_1^\dagger$ and $S_1^\dagger S_2^\dagger = -e^{i\theta} S_2^\dagger S_1^\dagger$ by multiplying $S_2^\dagger$ or similar operator both on left and right on both side of the equality. Then we have
\begin{align}
    4\{\re S_1, \re S_2\}=& (S_1+S_1^\dagger)(S_2+S_2^\dagger)+(S_2+S_2^\dagger)(S_1+S_1^\dagger)\nonumber\\
    =&(1-e^{-i\theta})S_1 S_2 +(1-e^{i\theta})S_2^\dagger S_1^\dagger + (1-e^{i\theta}) S_1 S_2^\dagger + (1-e^{-i\theta}) S_2 S_1^\dagger
\end{align}
Therefore,
\begin{equation}\label{eq:upperangle}
    4|\langle \{\re S_1, \re S_2\} \rangle| \leq 2\times |1-e^{-i\theta}|+2\times|1-e^{i\theta}| = 8|\sin(\theta/2)|
\end{equation}
The Schr\"{o}dinger uncertainty relation \cite{QM} is
\begin{align}\label{eq:uncertainty2}
    \var(A_1)\var(A_2)
    \geq \left| \langle A_1\rangle \langle A_2\rangle - \left\langle \{A_1, A_2\}\right\rangle/2  \right|^2 + \left|\left\langle[A_1,A_2]\right\rangle/i\right|^2
\end{align}
By taking $A_i = \re S_i$, for $i=1,2$, we have

\begin{align}
       \var(\re S_1)\var(\re S_2) \geq \left| \langle \re S_1\rangle \langle \re S_2\rangle - \left\langle \{\re S_1, \re S_2\}\right\rangle/2 \right|^2 .
\end{align}
In the case that $\langle \re S_1\rangle \langle \re S_2\rangle > |\sin(\theta/2)|$, Eq.~\eqref{eq:upperangle} implies that
\begin{align}
       \var(\re S_1)\var(\re S_2) \geq \left( \langle \re S_1\rangle \langle \re S_2\rangle-|\sin(\theta/2)|\right)^2.
\end{align}
Besides, in any case, it holds naturally that $\var(\re S_1)\var(\re S_2) \ge 0$.
\end{proof}

\subsection{2. Estimation of the fidelity for GHZ state}

The $d$-dimensional state $\ghz$ has the stabilizer group $\mathcal{S} = \{S_{abc}\}$, where
\begin{equation}\label{eq:stabilizer}
    S_{abc} = Z^aX^c \otimes Z^{(b-a)} X^c \otimes Z^{-b} X^c.
\end{equation}
We can verify that $S_{abc} = S_{a00} S_{0b0} S_{00c}$, and $\tr(S_{abc} S_{xyz}^\dagger) = d^3 \delta_{ax} \delta_{by} \delta_{cz}$.

It is known that, $
   |{\rm GHZ}\rangle \langle {\rm GHZ}| = \sum_{S\in \mathcal{S}}  S /d^3$.
Consequently, for a given state $\rho$, its fidelity $f$ with $\ghz$ can be written as
\begin{equation}\label{eq:fidelityexp}
f = \langle {\rm GHZ}|\rho|{\rm GHZ}\rangle = \frac{1}{d^3} \sum_{S \in \mathcal{S}} \tr( \rho S) = \frac{1}{d^3} \sum_{S \in \mathcal{S}} \re\expct{ S}_\rho, 
\end{equation}
where the last equality is from the fact that $S^\dagger \in \mathcal{S}$ if $S \in \mathcal{S}$. 

To estimate the fidelity $f$, we adopt the inflation process as described in Fig.~\ref{fig:GHZ}, and obtain a list of inequalities:
\begin{align}
 &|\langle S_{abc}\rangle_\rho|^2 +|\langle S_{xy0}\rangle_\eta|^2 \leq 1+ \left|\cos \left({ct\pi}/{d}\right)\right|,\label{eq:phase}\\
 &\var (\re S_{abc}')_\eta \var(\re S_{xy0})_\eta \geq [R(\langle \re S_{abc}'\rangle_\eta \langle \re S_{xy0}\rangle_\eta - \left|\cos \left({ct\pi}/{d}\right)\right| )]^2,\label{eq:cond2}\\
    &|\langle S_{xy0} \rangle_\eta| \geq  \langle \re S_{xy0}\rangle_\eta \geq  \langle \re S_{x'y'0} \rangle_\rho + \langle \re S_{(x-x')(y-y')0}\rangle_\rho -3/2,\label{eq:cond3}\\
    &|\langle S_{xy0}\rangle_\eta|^2 \geq |\langle S_{x'y'0} \rangle_\rho |^2 +|\langle S_{(x-x')(y-y')0} \rangle_\rho|^2 - {3}/{2},\label{eq:3/2}\\
    &|\langle S_{xy0}\rangle_\eta | \geq \re (\langle S_{xy0}\rangle_\eta) \geq 4f-3,\label{eq:3/4}
\end{align}
where $t$ is an arbitrary integer, and we have made use of the marginal relations in the inflation process, $c\neq 0$, $x=y\neq 0$ or only one of $x, y$ is $0$, $x'=x,y'=0$ in the case that $x=y\neq 0$, and $x'=y'$ is the nonzero element of $x, y$ in the case that only one of $x, y$ is $0$. Above we have used the marginal relations as follows:
\begin{align}
    \meane{S_{x00}}_\gamma=\meane{S_{x00}}_\rho,\\
    \meane{S_{0y0}}_\gamma=\meane{S_{0y0}}_\rho,\\
    \meane{S_{xy0}}_\eta=\meane{S_{xy0}}_\gamma,\\
    \meane{S_{abc}'}_\eta=\meane{S_{abc}}_\rho,
\end{align}
where $S_{abc}'$ is supported on $A,B$ and $C'$ similarly. Equation.~\eqref{eq:phase} is from Lemma~\ref{lm:incompatible}, Eq.~\eqref{eq:cond2} is from Lemma~\ref{lm:uncertainty}, Eq.~\eqref{eq:cond3} and Eq.~\eqref{eq:3/2} are from Lemma~\ref{lm:corr_sum}, and Eq.~\eqref{eq:3/4} is from Lemma~\ref{lm:fidelity}.
Besides, $\langle (\re S)^2\rangle \leq 1$ holds for any $S \in \mathcal{S}$. Note that Eqs.~\eqref{eq:cond3}--~\eqref{eq:3/4} can be used to impose additional inequalities among different components $\langle S_{xy0}\rangle$ for $x,y \in \{0,1,\cdots, d-1\}$. This would impose a tighter bound on the fidelity. Note that above all $a,b,c,x,y, e$ ranges from 0 to $d-1$, so there are only finite numbers of constraints.

Under those conditions and the fact in Eq.~\eqref{eq:fidelityexp}, we can find out an upper bound of fidelity $f$. For $d=2 ,3, 4, 5, 6, 7, 8$, the upper bound is $0.893, 0.950, 0.881, 0.925, 0.899, 0.912 , 0.876$, respectively.
We remark that, in the case that the dimension $d$ is an odd prime number, we can derive a simple upper bound by combining Eq.~\eqref{eq:phase} and Eq.~\eqref{eq:3/4}. Those two equations imply
\begin{align}
    &(4f-3)^2+|\expct{S_{abc}}_\rho|^2 \leq 1+ \sin \frac{\pi}{2d},\ c\neq 0,
\end{align}
which leads to
\begin{equation}\label{eq:fidelityprime}
f =\frac{1}{d^3}\sum_i \langle S_i\rangle \leq \frac{1}{d^3}\left[d^2+(d^3-d^2)\sqrt{1+\sin\frac{\pi}{2d}-(4f-3)^2}\right],
\end{equation}
since there are $(d^3-d^2)$ terms in the form $S_{abc}$ with $c\neq 0$, and $|\expct{S_{abc}}| \le 1$ in any case.

\begin{table}[h!]
    \centering
    \begin{tabular}{|c|c|c|c|c|}
    \hline
         $d$ & 2 & 3 & 4 & 5 \\ \hline
         Eq.~\eqref{eq:ghnetfidelity_app} & 0.9 & 0.955 & 0.9 & 0.935 \\ \hline
         Eq.~\eqref{eq:fidelityprime} & 0.9 & 0.951 &  & 0.925\\ \hline
         Numerical & 0.893 & 0.950 & 0.881 & 0.925\\ \hline
    \end{tabular}
    \caption{Upper bound on fidelity with different methods {, see the text for further details. The case $d=4$ with second method is empty since Eq.~\eqref{eq:fidelityprime} works for only prime dimensional case.}}
    \label{table:fidelities}
\end{table}
Now we go to compare the upper bound  in Eq.~\eqref{eq:ghnetfidelity} in main text, the upper bound given by Eq.~\eqref{eq:fidelityprime}, and the one derived from Eq.~(\ref{eq:phase}-\ref{eq:3/4}) for the GHZ states in dimension $d=2,3,4,5$. The result is summarized in Table~\ref{table:fidelities}.  {The case $d=4$ with second method is empty since Eq.~\eqref{eq:fidelityprime} works for prime dimension.} As we can see, those upper bounds are not so far away from each other.

\twocolumngrid

\bibliography{refs.bib}

\clearpage

\end{document}